\documentclass[twocolumn,preprintnumbers,amsmath,amssymb, superscriptaddress, showkeys,nofootinbib]{revtex4-1}
\DeclareMathOperator{\tr}{Tr}
\usepackage{color}
\usepackage{xcolor}
\usepackage{graphicx}
\usepackage{dcolumn}
\usepackage{bbm}
\usepackage{amsmath}
\usepackage{amsthm}
\usepackage{dsfont}
\usepackage{array}
\usepackage{wrapfig}
\usepackage{hyperref}
\usepackage[all]{hypcap}
\usepackage[latin10]{inputenc}
\usepackage{amsfonts}
\usepackage[english]{babel}
\usepackage{epstopdf}

\newcommand{\ket}[1]{|#1\rangle}

\newcommand{\ketbra}[2]{|#1\rangle\!\langle#2|}
\newcommand{\id}{{\mathbbm 1}}

\newcommand{\mdag}{^{\dag}}

\newcommand{\hsp}[1]{\hspace{#1 em}}
\newcommand{\sqz}{\hsp{-0.1}}
\newcommand{\nketbra}[2]{\left\vert{#1}\right\rangle \sqz\sqz\sqz \left\langle{#2}\right\vert}

\newtheorem{prop}{Proposition}

\newtheorem{defin}{Definition}
\newtheorem{lem}{Lemma}

\DeclareMathOperator{\rank}{rank}

\DeclareMathOperator{\SE}{SE}

\begin{document} 
%\author{J. M. Matera$^1$ $^2$\footnote[1]{These authors contributed equally to this work}, D. Egloff$^1$ \footnotemark[1] , N. Killoran$^1$ and M. B. Plenio$^1$ }
%\address{$^1$ Institut f\"ur Theoretische Physik and IQST, Albert-Einstein-Allee 11, Universit\"at	Ulm, D-89069 Ulm, Germany}
%
%\address{$^2$ {Departamento de F\'isica-IFLP, Universidad Nacional de La Plata, C.C. 67, La Plata 1900, Argentina}}
\title{Coherent Control of Quantum Systems as a Resource Theory}

\author{J. M. Matera}\thanks{These authors contributed equally to this work}
\affiliation{Institut f\"ur Theoretische Physik and IQST, Albert-Einstein-Allee 11, Universit\"at
	Ulm, D-89069 Ulm, Germany}
\affiliation{Departamento de F\'isica-IFLP, Universidad Nacional de La Plata, C.C. 67, La Plata 1900, Argentina}

\author{D. Egloff}\thanks{These authors contributed equally to this work}
\affiliation{Institut f\"ur Theoretische Physik and IQST, Albert-Einstein-Allee 11, Universit\"at
	Ulm, D-89069 Ulm, Germany}

\author{N. Killoran}
\affiliation{Institut f\"ur Theoretische Physik and IQST, Albert-Einstein-Allee 11, Universit\"at
	Ulm, D-89069 Ulm, Germany}

\author{M. B. Plenio}
\affiliation{Institut f\"ur Theoretische Physik and IQST, Albert-Einstein-Allee 11, Universit\"at
	Ulm, D-89069 Ulm, Germany}

\begin{abstract}
	Control at the interface between the classical and the quantum world is fundamental in quantum physics. In particular, how classical control is enhanced by coherence effects is an important question both from a theoretical as well as from a technological point of view. In this work, we establish a resource theory describing this setting and explore relations to the theory of coherence, entanglement and information processing.
Specifically, for the coherent control of quantum systems the relevant 	resources of entanglement and coherence are found to be equivalent and closely  related to a measure of discord. The results are then applied to the DQC1 protocol and the precision of the final measurement is expressed in terms of the available resources.
\end{abstract}
\maketitle

{\bf \em Introduction ---\phantomsection{}\addcontentsline{toc}{section}{Introduction}}  Coherent superposition is a defining characteristic of the quantum world.
Coherence indicates the fundamental misalignment, or noncommutativity,
between quantum states and the interactions or observables which we may
use to probe them. Due to its intimate connection with quantum superposition,
coherence is also important in a large number of quantum information protocols.
In fact, coherence can be seen as a type of \emph{resource}, allowing one
to perform tasks which would be more difficult or not possible otherwise.
Indeed, coherence has recently been developed into a formal quantum resource
theory \cite{BaumgratzCP14,WinterY2015,Yuan2015,DuBQ2015_maxco,xi2014quantum,KilloranSP2015,Streltsov2015,YaoXG+2015,chitambar2015assisted,chitambar2015relating,streltsov2015hierarchies,hu2015,hu2015coherence,ma2015converting}
\footnote{Another possible approach to coherence theory \cite{marvian2014,aaberg2014}
is based on the theory of reference frames \cite{bartlett,GourS08}, which proved
useful in quantum thermodynamics\cite{lostaglio2014}.}, similar to that for
entanglement \cite{plenio2007introduction,HorodeckiHHH2009}.

In the macroscopic classical world, where states and observables commute, superposition
effects are suppressed and physical systems can be described without coherence, using
classical probability distributions. Yet some special systems, often found at mesoscopic
scales, can exist in the murky borderlands between the classical and quantum worlds. In
fact, systems which bridge between these worlds are very important in
modern experiments. Operationally, it is common to employ intermediary physical systems,
such as lasers, magnetic fields, or photodiodes, to interface with a separate ``target''
quantum system. By coupling to the target system, these mediator systems can function
as state preparation, control, and measurement devices.

In order to interact meaningfully with the controlled system, the mediator systems must
themselves be able to exert a nonclassical effect on their targets. At the same
time, they must also interface with the classical world in order to communicate human- or machine-readable 
instructions and measurement outcomes. Through this, they are inevitably exposed to
classical noise and decoherence effects which makes the creation and the conservation of coherence
a costly task. Recognising that coherence is a potential resource, we might ask: 
what value might be gained if we were to pay these costs? What potential quantum advantages 
do coherent resources provide in this standard operational paradigm?

In this work, we address these questions by formalizing a resource theory for the tasks
of preparing, controlling, and measuring quantum systems, and explore the differences
between having incoherent versus coherent resources at our disposal. Within this framework,
coherence and entanglement can be freely interconverted and thus represent the same
underlying resource, which we call the \emph{recoverable coherence}. We introduce a
quantifier for this resource, and connect it to measures of quantum discord\cite{OZ.01,vedral_discord,Modi2012,Streltsov2015a}. 
Finally, we illustrate these ideas through an application to the family of quantum algorithms known as ``Deterministic Quantum Computation with one qubit'' (DQC1) proposed in \cite{Knill}. In the last decade this family of algorithms   has instigated a lively debate as to  what is the 
quantum resource behind the speed up obtained with quantum algorithms, since DQC1 can be implemented even with a very small amount of entanglement\cite{ParPle.02,Datta.05,DSC.08,lanyon,Ali.2014,ma2015converting}.   
We show how  the accuracy of the outcome in DQC1 can be quantified in terms of the recoverable coherence and 
how this connects with entanglement and discord. We discuss connections with related works
in the section defining the resource theory and in the \hyperref[appendix]{Appendix}.
\begin{figure}[t]	
	\centering
	\scalebox{.5}{\includegraphics{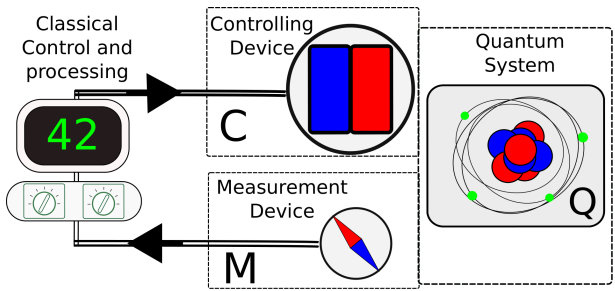}}
	\caption{ {\bf  Devices between the quantum and the classical.}\label{fig:classical-quantum-control} As we move more from the left to the right hand side, the scale of the system reduces and we have access to a stronger coherent control.
	}
\end{figure}
\\

{\bf \em The framework ---\phantomsection{}\addcontentsline{toc}{section}{Framework}} To motivate the following framework, we start by considering 
a generic experimental setup for controlling a quantum system (see Fig. 
\ref{fig:classical-quantum-control}). Humans can only interact mechanically 
with macroscopic objects, therefore one part of any experiment needs to be 
macroscopic, be it only the keyboard of a computer. This part is fully
described by classical physics.
To model this in quantum mechanics, following~\cite{BaumgratzCP14} we say
that the state of a system is incoherent if it is  diagonal in a fixed basis.
\begin{defin}
    Given a system $A$ and a fixed orthonormal basis ${\cal Z}=\{\ket{c}\}_{c=0}^{D-1}$,
    we call a state $\rho$  \emph{incoherent} (with respect to $\cal Z$) if $\rho=\sum_c p_c |c\rangle\langle c|$
    for some $p_c\geq 0$, where $\sum_c p_c=1$.
\end{defin}

Secondly, there are some experimental devices which allow us to operate on
the full quantum system. These intermediary devices are usually in the mesoscopic 
domain, since we want to manipulate their operating parameters
in a deterministic way and, at the same time, use them to manipulate microscopic quantum system. Because of their size it is an operationally hard task
to bring these devices controllably into a coherent superposition that
remains stable against decoherence. 
In order to cast this situation into a resource theory perspective, it is
advantageous to assume that only incoherent operations are available and all 
required coherence is supplied by a third party. For this 
reason, we might think of coherence as a resource for the manipulation of the controlling devices.
Following~\cite{BaumgratzCP14}, we thus define\footnote{See the \hyperref[incc]{Appendix} for a brief discussion on this choice.}:
\begin{defin}\label{def:inc}
	We call a quantum operation incoherent if each of its Kraus operators $K_{\alpha}$ is incoherent.
    That is, for any $\sigma$ incoherent, $K_{\alpha} \sigma {K_{\alpha}}\mdag$ is incoherent.
\end{defin}
We want to use these intermediary devices to \emph{control} a quantum system, and in
the best case, we can have perfect quantum control. We must keep in mind that the
controlling devices decohere quickly due to their size. Therefore it might be hard to
encode information in a control basis which is not the incoherent one and still
have a stable quantum control.
The best stable control we can hope for is therefore given by the unitary
\begin{align}\label{eq:control}
U_{control}=\sum_{c} \ketbra{c}{c} \otimes {U_c}.
\end{align}
We also assume that we can add ancillae to the quantum system and that we can
trace out parts of the system, i.e. we can \emph{prepare} and \emph{discard}
parts of it.

Ultimately we might wish to \emph{measure} the quantum system. We therefore
include a measurement device, which couples to the quantum system, performs
any measurement of the quantum system, and sends the classical measurement
result to a computer. This is described by
\begin{equation}\label{eq:measure}
    \ketbra{0}{0}\otimes\rho \mapsto \sum \limits_c \ketbra{c}{c}\otimes K_c\rho K_c\mdag,
\end{equation}
where the left-hand side of the tensor product denotes a register in the
computer's memory and the right hand side is the state of the quantum system.

We note that incoherent operations include anything one can do with a computer.
Since in this paper we focus on the control of a quantum system,
we only need to consider two systems: the controlling system A, on which we can
do any incoherent operations for free, and the quantum system B, for which we
have full quantum control, including measurements. In this way, the complete
family of allowed operations in our framework is defined as follows:
\begin{defin}~\label{def:framework}
	Consider a bipartite system $AB$. The class of \emph{Global Operations Incoherent
    on $A$}, with respect to the local orthonormal basis  ${\cal Z}=\{|c\rangle\}$ on
    $A$ (abbreviated as $GOIA_{\cal Z}$), is the family of quantum channels consisting
    of (finite) combinations of:
	\begin{enumerate}
		\item Incoherent operations on $A$ (Def. \ref{def:inc})
		\item Controlled operations in the incoherent basis from $A$ to $B$ (Eq.~\ref{eq:control})
		\item Adding or removing (tracing out) ancillae on $B$
		\item Measurement and postselection on $B$ (Eq.~\ref{eq:measure}).
	\end{enumerate}
\end{defin}
\noindent Notice that if we extend this set of operations by allowing general unitary
operations on $A$, we recover the full set of quantum operations on $AB$.\\

{\bf \em The Resource Theory ---\phantomsection{}\addcontentsline{toc}{section}{Resource Theory}}
Having defined the free operations, we need to address what are the free and the resource
states in the framework, and investigate how resources can be distilled if someone provides
a source of non-free states.

{\em Free states ---\phantomsection{}\addcontentsline{toc}{subsection}{Free states}} With $\mathcal{Z}=\{\ket{c}\}$ as the incoherent basis of A, we can prepare any state
of the form $\rho_0=\sum_c p_c\ketbra{c}{c}_A\otimes\ketbra{0}{0}_B$
using only incoherent operations on $A$. Performing controlled operations (Eq.~\ref{eq:control}) on $\rho_0$ (with the aid of ancillary states), we can prepare any state in the set
of \emph{$\mathcal{Z}$-classical-quantum states} \cite{chitambar2015assisted},
\begin{equation}
    CQ_{\cal Z}:=\{\rho~|~\rho=\sum_c p_c|c\rangle\langle c | \otimes \rho_c,~\ket{c}\in\mathcal{Z}\},
\end{equation}
where the $\rho_c$ are arbitrary quantum states. Conversely, any operations in the
$GOIA_{\cal Z}$ framework conserve this set. The largest set of 
operations that preserves the set of classical-quantum states $CQ_{\cal Z}$ was 
defined in \cite{ma2015converting}. While all of our operations are inside that 
set, the converse remains an open question.
Using the physical picture we have introduced, we can link together other seemingly
disparate recent works in the field~\cite{Streltsov2015,YaoXG+2015,chitambar2015assisted,chitambar2015relating,streltsov2015hierarchies,hu2015,hu2015coherence,ma2015converting} (see the \hyperref[appendix]{Appendix} for a brief overview
and also see the related independent work~\cite{Yadin2015},which discusses a different subset of the $CQ_{\cal Z}$-preserving operations~\cite{ma2015converting}).

Most notably, in~\cite{chitambar2015assisted} the subset of the bipartite 
operations on $AB$, that can be performed \emph{locally}, only with the aid of 
classical communication (the Local Quantum-Incoherent operations and Classical
Communication, $LQICC_{\cal Z}$) was introduced. These operations are a strict subset of $GOIA_{\cal Z}$: we 
get them by restricting the control to be performed by local operations on $B$, conditional to measurementes outcomes on $A$. The connections will prove useful to unravel the resource 
theory defined by $GOIA_{\cal Z}$, which is done in the following sections.

{\em Coherence as a resource ---\phantomsection{}\addcontentsline{toc}{subsection}{The resources}}
Now we determine the set of resource states in our framework, i.e. the set of
states which makes $GOIA_\mathcal{Z}$ operations universal. To this end
we first note that from a supply of maximally coherent states $|+\rangle$ on $A$ we
are able to implement any local operation on $A$ \cite{BaumgratzCP14}. Secondly, the
supply of $|+\rangle$ on $A$ also allows for the generation of
entanglement between $A$ and $B$ by application of Eq.~\ref{eq:control}. Thirdly, the
provision of arbitrary local operations on $A$ and $B$ and arbitrary amounts of entanglement
between $A$ and $B$ allows for the generation of arbitrary joint operations between
$A$ and $B$ \cite{EisertJP+2000} (see also \cite{streltsov2015hierarchies}).
In particular, by applying a CNOT (included in the $GOIA_{\cal Z}$-, but not in the $LQICC_{\cal Z}$-operations) one can create a pure maximally entangled state (a singlet) from a maximally coherent state~\cite{Streltsov2015} and one can steer an incoherent state to a maximally coherent one under $LQICC_{\cal Z}$ operations by using up a singlet state~\cite{hu2015}. Therefore the pure resource states can be produced from one another with $GOIA_{\cal Z}$ operations.

We can now ask how many resource states one can distill from $n$ copies of a given state.
\begin{defin}
	Let $r^{\epsilon}(\rho,n)\cdot n$ be the maximal number of fully coherent qubit states
    (on subsystem $A$) that can be prepared from $n$ copies of the state $\rho_{AB}$
    with fidelity at least $1-\epsilon$, by applying maps $\Lambda \in  GOIA_{\cal Z}$:
	$$
	r^{\epsilon}(\rho,n)
	:= \sup\limits_{\Lambda \in GOIA_{\cal Z}}\!\left\{R\;| \;F\left(\Lambda(\rho^{\otimes n}),
    \ketbra{+}{+}^{\otimes n R}\right)\geq 1- \epsilon\right\}.
	$$
	The \emph{recoverable coherence} (with respect to the basis $\cal Z$) is the infinite-copy
    and infinitesimal error limit of the above maximal ratio:
	\begin{equation}
	C^{REC}_{{\cal Z}}(\rho) = \lim_{\epsilon\rightarrow 0}\lim_{n\rightarrow \infty}  r^{\epsilon}(\rho,n).
	\label{eq:defrrec}
	\end{equation}
\end{defin}

As entangled and coherent resource states can be interconverted, the analogous notion of
recoverable entanglement coincides with $C_{\cal Z}^{REC}$:
\begin{equation}
    E_{{\cal Z}}^{REC}(\rho)\equiv C_{{\cal Z}}^{REC}(\rho).
\end{equation}
Notice that these quantities are not equivalent to the distillable entanglement\cite{plenio2007introduction}. Moreover, they are not entanglement monotones \cite{vidal200} since $GOIA_{\cal Z}$ allows to convert product states (not incoherent on $A$) into entangled states.
However, for pure states, distillable entanglement is a lower bound to the distillable
entanglement under $GOIA_{\cal Z}$
\footnote{For pure states, the protocol described in~\cite{bennett_pure_entanglement_distillation} for entanglement
distillation can be performed, since it requires full quantum control just on one side.
This shows that the entanglement entropy provides a lower bound.}. 
We give general lower bounds to the distillable coherence in the \hyperref[lower_bounds]{Appendix}.
As was noticed in \cite{chitambar2015assisted}, for any state outside of $CQ_{\cal Z}$,
there is a protocol in $LQICC_{\cal Z}\subset GOIA_{\cal Z}$ which allows to recover some amount of coherence. Hence, there is no bound coherence or entanglement in $GOIA_{\cal Z}$.

\noindent{\em A monotone for $C_{\cal Z}^{REC}$ --- \phantomsection{}\addcontentsline{toc}{subsection}{A monotone}}
A next natural step in the resource theory is to introduce a monotone which
quantifies the distance to the free states. A particularly suitable measure
is provided by the relative entropy:
\begin{equation}
\label{def:distcqia}
\Delta_{\cal Z}(\rho)=\min_{\sigma \in CQ_{\cal Z} } S(\rho||\sigma)=S(\rho')-S(\rho),
\end{equation}
\noindent where $S(\rho||\sigma)={\rm Tr} [\rho(\log_2(\rho)-\log_2(\sigma))]$,
$S(\rho)=-{\rm Tr}[\rho \log_2\rho]$ is the von Neumann entropy and
$\rho'=\sum_c (|c\rangle\langle c|\otimes \id)\rho (|c\rangle\langle c|\otimes \id)$
is the state obtained from $\rho$ by completely decohering with respect to the basis
$\mathcal{Z}$ (see Lemma 1 in the \hyperref[proofs]{Appendix}).

The functional $\Delta_\mathcal{Z}$ is additive, convex and monotonic under
$GOIA_{\cal Z}$ on average. Proofs of these properties are presented in the \hyperref[proofs]{Appendix}
in Lemma 2, 3 and Proposition 1.
Most importantly, it upper bounds $C_{\cal Z}^{REC}$ via (see the \hyperref[upper bound]{Appendix}
for a proof).
\begin{equation}\label{thm:geomupperbound}
    \Delta_{{\cal Z}}(\rho)\geq   C_{{\cal Z}}^{REC}(\rho).
\end{equation}
This sharpens a similar result of \cite{chitambar2015assisted} where the free
operations defining $C_{\cal Z}^{REC}$ on the right hand side were the more
restrictive $LQICC_{\cal Z}$ operations. Still, since $LQICC_{\cal Z}$ is strictly included in $GOIA_{\cal Z}$, the bounds derived in~\cite{chitambar2015assisted,streltsov2015hierarchies} are valid for our framework and we get that Eq.~\ref{thm:geomupperbound} is tight for maximally correlated states ($\rho=\sum_{cc'} \rho_{cc'} |c\rangle\langle c'|\otimes |c\rangle\langle c'|$~ \cite{Rains.99,HH.01}) and for product states~\cite{chitambar2015assisted} (where the recoverable coherence is just the distillable coherence on the $A$-part calculated in~\cite{WinterY2015}), as well as for general pure states~\cite{streltsov2015hierarchies} (as the proofs simplify in our framework, we show them nonetheless in the \hyperref[lower_bounds]{Appendix}). The bound is also tight for quantum-classical states. We have to leave open the question of whether the bound can be reached in general.\\

{\bf \em Basis-independent recoverable coherence and discord ---\phantomsection{}\addcontentsline{toc}{section}{Coherence and Discord}}
In the previous section, we have presented a framework which specifies coherence
in some fixed basis $\mathcal{Z}$. However, it might be useful in some contexts
to work without this constraint. The natural extension would then be to ask
what is $C_{\cal Z}^{REC}$ of the state in the most unfavourable case\footnote{At a first glance,
the maximum recoverable coherence might also seem a meaningful quantity. However, it depends strongly on the dimensionality
of the basis, and would be saturated for any pure state by choosing as the incoherent basis the one conjugate to the local Schmidt basis.
}.
For product states, it is clear that the choice of the eigenbasis of $\rho_A$
is the worst case. In general, this does not need to be the case. For instance,
consider the state
$$ \rho=\epsilon \nketbra{\uparrow}{\uparrow}\otimes\nketbra{\uparrow}{\uparrow}
+ \tfrac{1-\epsilon}{2} \left(\nketbra{\leftarrow}{\leftarrow}\otimes\nketbra{\uparrow}{\uparrow} +\nketbra{\rightarrow}{\rightarrow}\otimes\nketbra{\downarrow}{\downarrow}\right).$$
In the limit $\epsilon \rightarrow 0$, $C^{REC}_{\cal Z}\rightarrow 1$ using the local eigenbasis ${\cal Z}=\{\ket{\downarrow},\ket{\uparrow}\}$, yet $C^{REC}_{\cal Z'}\rightarrow 0$ for the choice ${\cal Z}'=\{\ket{\leftarrow},\ket{\rightarrow}\}$. With this in mind, we define the
\emph{basis-independent recoverable coherence}
\begin{equation}\label{def:DC}
    C^{REC}_{min}[\rho]=\min_{{\cal Z}}C^{REC}_{{\cal Z}}(\rho),
\end{equation}
as the minimum recoverable coherence regarding the most unfavourable basis. Due to
Eq.~\ref{thm:geomupperbound}, we obtain
\begin{equation}
    C^{REC}_{min}[\rho] \leq  \min_{\cal Z}\Delta_{\cal Z}(\rho) = \Delta^{A \rightarrow B} (\rho).
\end{equation}
That is, the basis-independent recoverable coherence is upper bounded by the thermal discord
$\Delta^{A \rightarrow B} (\rho)$ (also called one-way information deficit). Thermal discord
represents the difference between the work that can be extracted from a system in the state
$\rho$ by performing either global or local operations on $A$~\cite{Zurek.03,Horodeckies.05}.
Additionally, it corresponds to a particular case of a measure of \emph{discord}, a type of
non-classicality of a quantum state beyond the notion of entanglement
\cite{OZ.01,vedral_discord,Modi2012,Streltsov2015a}, quantifying how much a given state
fails to belong to the set of Classical-Quantum or pointer states \cite{OZ.01,DVB.10}:
\begin{equation}
CQ:= \bigcup_{\cal{Z}} CQ_{\cal{Z}} .
\end{equation}
See Fig.~\ref{fig:sets} for a picture of the relevant sets.
The most prominent features of discord quantifiers are~\cite{Modi2012} their
\begin{itemize}
	\item Vanishing, iff the state is in $CQ$, and
	\item Invariance under local unitaries.	
\end{itemize}
Therefore discord quantifiers are asymmetric with respect to the swap of $A$ and $B$.
One may also ask that a discord quantifier is nonnegative, bounded from above by the
entropy of $\rho_A$ and suitably normalized, such that the measure coincides with the
entanglement for the singlet state in the qubit case.
We note that the basis-independent recoverable coherence $C^{REC}_{min}$ is also a discord quantifier since it has the above mentioned properties.
For pure states, distillable entanglement is the same as thermal discord
\cite{bennett_pure_entanglement_distillation}, which is an upper bound for the
basis-independent recoverable coherence, which in turn is an upper bound for
entanglement as noted above. So, for pure states these quantities all coincide
(also see~\cite{chitambar2015assisted} for a different argument).
\begin{figure}[h]
	\centering
	\scalebox{.5}{\includegraphics{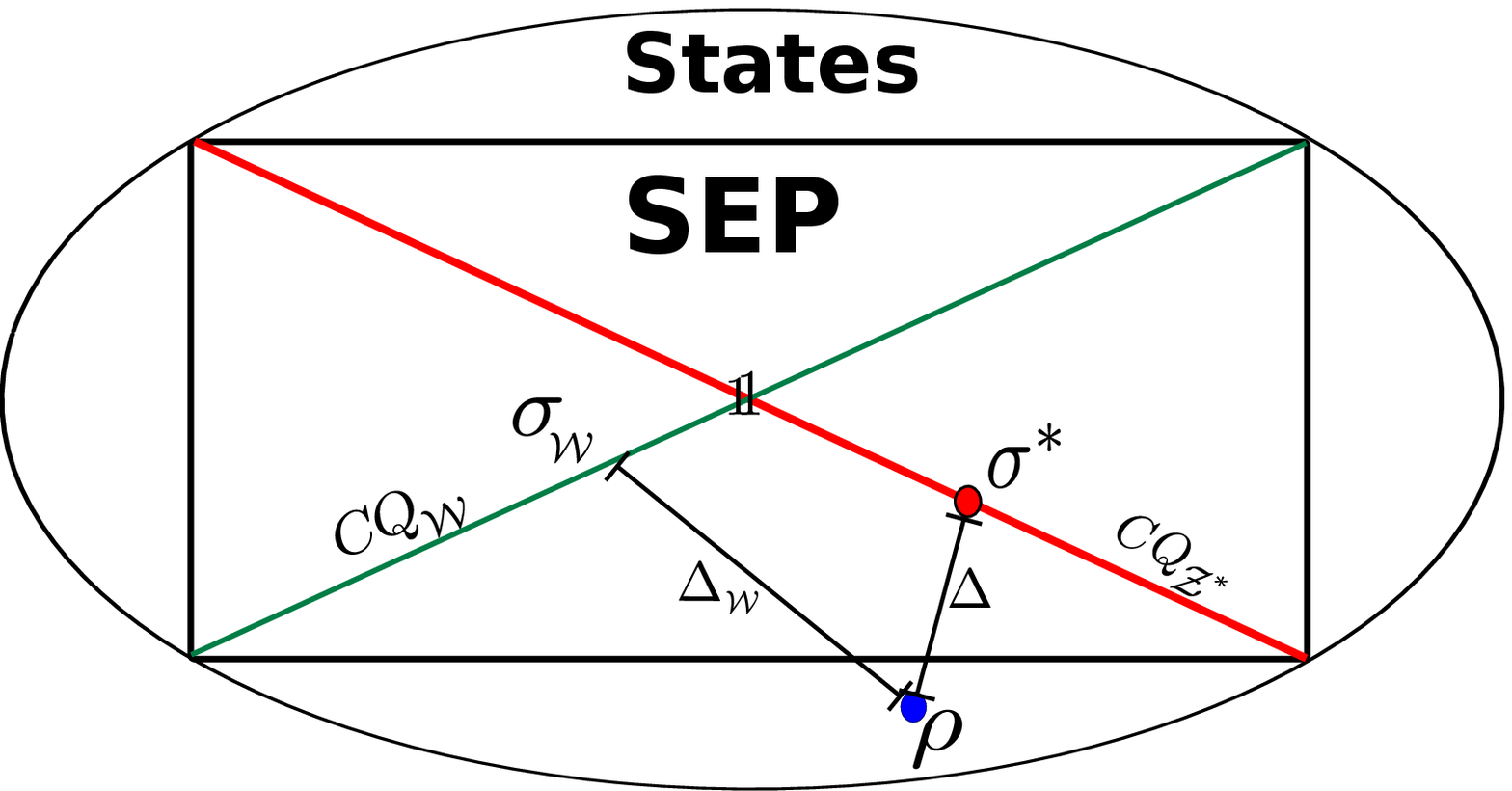}}
	\caption{{\bf
  Scheme of the discussed sets.} The set of zero discord states, $CQ$, is the union of all $CQ_{\cal Z}$ sets (e.g. $\cal Z=\cal W$). Each of the $CQ_{\cal Z}$ sets is convex, but $CQ$ is not convex. The union $CQ$ is contained within  the convex  hull of  the $CQ_{\cal Z}$, which is the set of separable states. The intersection of all $CQ_{\cal Z}$ are the states of the form $\frac{\id}{\tr \id} \otimes \rho_B$. The corners of the separable set correspond to pure states, which are shared with (in principle, many) incoherent lobes. $\cal Z^{\ast}$ is the basis in which the geometric distance $\Delta_{\cal Z}$ to $CQ_{\cal Z}$ gets minimized, i.e. one with $\sigma_{\cal Z}$ nearest  to $\rho$ in relative entropy.}
	\label{fig:sets}
\end{figure}

{\bf \em Application to the DQC1 protocol ---\phantomsection{}\addcontentsline{toc}{section}{DQC1}}
In this section we apply the above mentioned results to analyse the resources 
involved in the DQC1 protocol~\cite{Knill,DSC.08}. The goal of DQC1 is 
to determine the trace of a n-qubit unitary operator ${\bf U}$, which is a 
very challenging task in the realm of classical physics. The DQC1 protocol accomplishes this task 
by making use of a maximally coherent control qubit 
$\frac{\ket{0}+\ket{1}}{\sqrt{2}}$
as a probe and a maximally mixed state 
on the remaining target system 
(see Fig.~\ref{fig:dqc1}). 
After the action of the controlled unitary 
$\ketbra{0}{0}\otimes \id + \ketbra{1}{1}\otimes U$, 
the state of the probe encodes the trace of the unitary in the coherent bases 
($\{\frac{\ket{0}+\ket{1}}{\sqrt{2}},\frac{\ket{0}-\ket{1}}{\sqrt{2}}\}$ and $\{\frac{\ket{0}+i \ket{1}}{\sqrt{2}},\frac{\ket{0}-i \ket{1}}{\sqrt{2}}\}$). 
Measuring the probe in these bases ends the protocol. Of course, to read  out 
the result 
we need to perform repeated measurements of the final state, implying that we need to repeat the protocol many times to gain a certain degree of accuracy. If the initial state is not maximally coherent, we are still able to perform the algorithm, but we need a larger number of runs to reach the same precision~\cite{lanyon}. Interestingly, the protocol remains efficient also in the case of probes in a highly mixed state, even when the bipartite entanglement between the probe and any part of the system is  small or even vanishes during the entire protocol.
\begin{figure}[h]
	\centering
	\scalebox{.25}{\includegraphics{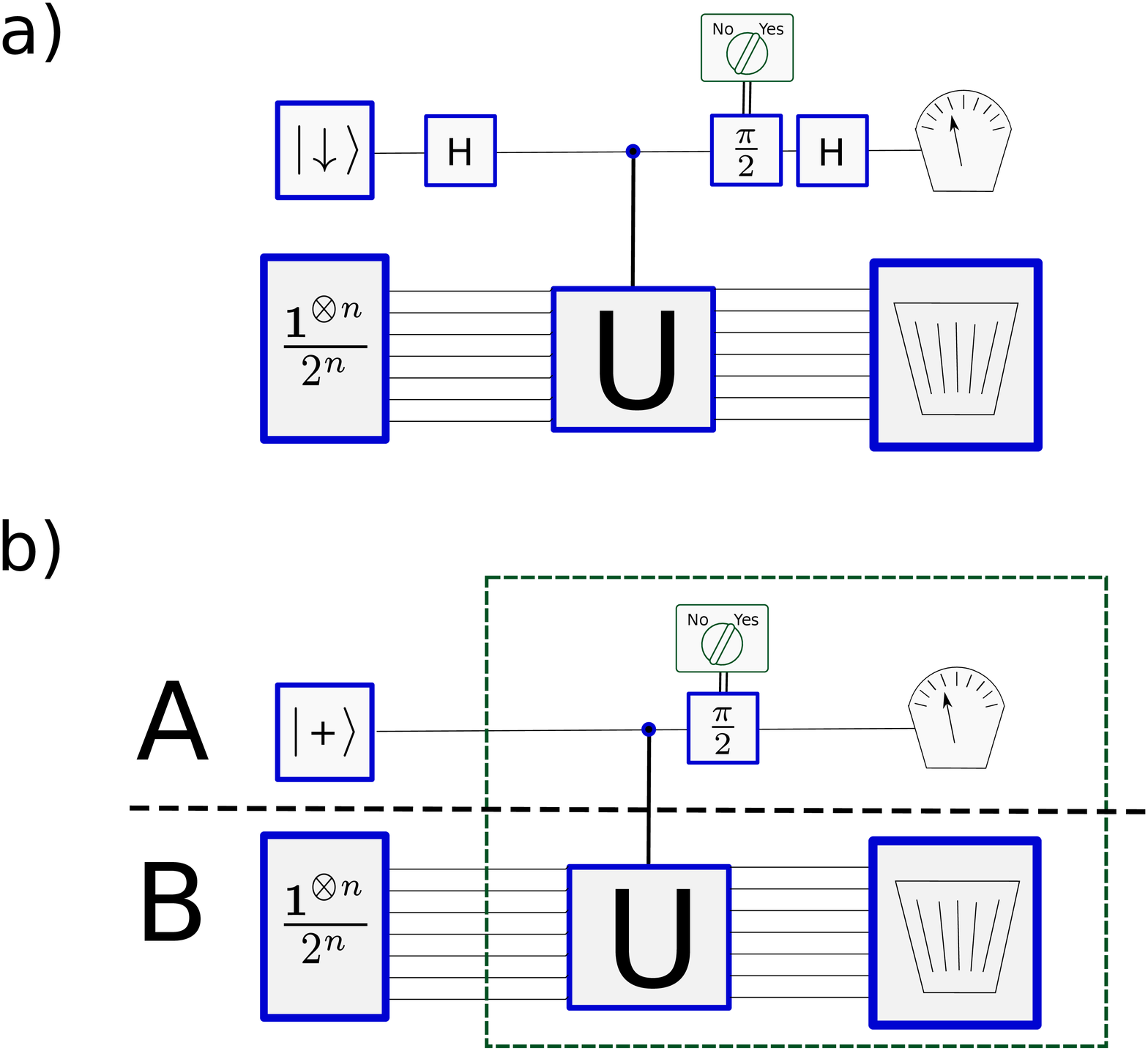}}
	\caption{{\bf Scheme of DQC1 protocol.} Top: the standard form of the protocol. Bottom: The protocol from the point of view of the incoherent-quantum partition. In this case, the first Hadamard gate was replaced by a probe (system A) in a fully coherent state($|+\rangle$), while the final Hadamard gate was exchanged by a final
		destructive measure in a maximally coherent basis (e.g. the basis $\{\ket{+},\ket{-}\})$.}
	\label{fig:dqc1}
\end{figure}
\noindent This observation provided motivation to look 
at different measures of quantumness, such as multipartite entanglement~\cite{ParPle.02} and quantum discord~\cite{DSC.08,ma2015converting}. But while it is not clear why one should look for multipartite entanglement in a setting that physically is bipartite, quantum discord is problematic as a resource since the zero discord set is not convex and thus mixing two zero discord states (which amounts to forgetting which of the two one prepared) can provide you some non-zero discord state (see Fig~\ref{fig:sets}). 
On the other hand, as it was pointed out in~\cite{ma2015converting}, the minimal 
requirement for the DQC1 protocol to work is the presence of some amount of 
coherence in the probe. We can make this statement more precise by remembering 
that, to obtain the expectation value which encodes $\mathrm{Tr}\mathbf{U}$, the 
protocol should be performed many times, consuming on each run a fresh 
qubit probe~\cite{DSC.08}. The number of runs needed to reach some desired precision 
depends directly on the degree of coherence of the probe. Suppose we have 
$m$ copies of the joint initial state {$\rho=\rho_0\otimes \id_{target}/\dim$}, where  $\dim$ is the dimension of the target system and $\rho_0$ is the (general) qubit state of the probe before applying the controlled unitary. 
We show in the \hyperref[appendix]{Appendix} that the precision (i.e. the number of binary significant digits) of the estimated ${\rm Tr \mathbf{U}}$ is (up to a constant) given by a function of $C^{REC}_{\cal Z}$, i.e.:
\begin{equation}\label{eq:precvscoh}
	prec(\tr {\bf U}) \approx -\log_2 |{\SE}\left( \frac{\tr {\bf U}}{\dim
	}\right)| \approx    \frac{1}{2}\log_2(C_{\cal Z}^{REC}(\rho_{0}))\;,
\end{equation}
\noindent where $\SE(\hat{x})$ denotes the \emph{standard error of the mean}\cite{devore.08} associated to the random variable $\hat{x}$.

Notice that in the present formalism, entanglement and coherence are interconvertible, so the amount of bipartite entanglement that can be produced during the protocol between the probe and any part of the target system is bounded by $C_{\cal Z}^{REC}(\rho)$, for any state $\rho$ of the total system at any stage of the protocol. We note that if at any point in the protocol any discord quantifier is non-zero, this implies that the state is not quantum-classical and therefore $C_{\cal Z}^{REC}$ is non-zero. By the monotonicity of $C_{\cal Z}^{REC}$ under $GOIA_{\cal Z}$, we find that the state of the probe at the beginning cannot have been incoherent. Therefore any discord quantifier is a witness for recoverable coherence and the applicability of the DQC1 protocol.\\

{\bf \em Conclusion ---\phantomsection{}\addcontentsline{toc}{section}{Conclusion}}
In this work, a framework for the description of incoherent systems 
controlling quantum systems was proposed. The set of operations over 
the composite system ($GOIA_{\cal{Z}}$), together with its associated 
minimal invariant set define a formal resource theory, in which
the resource is the amount of coherence that can be recovered on the 
control side.
Using the connections with other frameworks~\cite{Streltsov2015,YaoXG+2015,chitambar2015assisted,chitambar2015relating,streltsov2015hierarchies,hu2015,hu2015coherence,ma2015converting},
we extended many of the previous results to the $GOIA_{\cal{Z}}$ framework and showed that the associated resources, the recoverable coherence and recoverable entanglement, are equivalent. We upper bounded $C_{\cal Z}^{REC}$ by a geometric functional and found that the latter is a monotone of the theory.
By looking at the least favourable choice of the incoherent basis, the amount of resource associated to a given state is a discord quantifier. This quantifier is bounded from above by the thermal discord of the state. 
Finally, we exemplified our findings by
calculating the precision of the DQC1 protocol with a mixed control qubit and stated it in terms of the resource of our theory---the recoverable coherence.\\

%\ack
{\bf \em Acknowledgements ---}
We gratefully acknowledge discussions with Andrea Smirne and Antony Milne.
M.M. is supported by CONICET. D.E., N.K. and M.B.P. are supported by an Alexander 
von Humboldt Professorship, the ERC Synergy Grant BioQ and the EU project QUCHIP 
and EQUAM.
\vspace*{1cm}
{ \begin{center}
			{\bf  \large Appendix }\phantomsection{}\addcontentsline{toc}{part}{Appendix} \label{appendix}
		\end{center} }
{\bf \em Related works ---\phantomsection{}\addcontentsline{toc}{section}{Related works}}\label{recap}
One of the first papers relating coherence with entanglement was~\cite{Streltsov2015}, where they looked at how one can restrict coherence theory further by splitting the space of free states into two incoherent parts. This gives rise to control operations and thus allowed operations can produce entanglement by using up coherence. This line of thought was further developed in~\cite{YaoXG+2015}, where they showed an equality between symmetric discord~\cite{modi2010unified} and coherence in such a framework. Another approach was taken in~\cite{KilloranSP2015}, relating a more general form of superpositions than coherence to entanglement via control gates.  In~\cite{chitambar2015assisted} they instead looked at possible extensions of coherence theory and introduced the framework of local operations and classical communication ($LQICC_{\cal Z}$). This and modifications thereof where subsequently discussed in~\cite{streltsov2015hierarchies,chitambar2015relating}, which were mostly concerned with the relation of coherence to entanglement and structure of the respective theories. After these~\cite{hu2015,hu2015coherence} analysed steering induced coherence in the $LQICC_{\cal Z}$ framework and defined a measurement induced disturbance measure for coherence, to some extent related to discord. Finally ~\cite{ma2015converting} defined the set of quantum operations that preserve the form of incoherent-quantum states.
In a subsequent work~\cite{Yadin2015}, a subset of these operations was analysed independently of this letter.\\
{\bf \em Choice of the incoherent operations ---\phantomsection{}\addcontentsline{toc}{section}{Choice of incoherent operations} }\label{incc}
In our theory of coherent control of a quantum system, we need to choose which theory inside quantum mechanics corresponds best to a classical theory of labels, i.e. we need to specify the theory of coherence we choose to model the free operations on the "classical" side $A$ of our set-up. Clearly, we need to ensure that no coherence is produced, i.e. we only consider maps which, starting from an arbitrary incoherent state, result in yet another incoherent state (even after a possible post-selection).
This leads to the coherence theory defined in~\cite{BaumgratzCP14} and used here. 
One could now argue that one should restrict the allowed operations further. Indeed there is by now an entire hierarchy of  proposals of possible theories (see e.g. the appendix of~\cite{Streltsov2015b} for an overview), but the discussion of which is the most meaningful is far from settled yet (and context dependent). We adopted our current choice of model principally because it allows for a particularly
wide class of operations and hence will give particularly strict bounds in the sense that
a process that is impossible under this set will also be impossible under essentially all
other possible choices of coherence theory.

A totally different choice would be the theory of U(1)-covariance~\cite{marvian2014}, which has been successfully applied in the context of thermal operations (see e.g.~\cite{lostaglio2014}). The reason we do not use it here is simple: permutations are not allowed operations there (unless on degenerate subspaces), but since this only requires changing labels in our setting, this should be allowed. 

{\bf \em Proofs ---\phantomsection{}\addcontentsline{toc}{section}{Proofs}}\label{proofs} Here, the proofs of the properties presented in the main text are shown.
\begin{lem}[{Minimal relative entropy to $CQ_{\cal Z}$~--- } Eq. 6. Also see~\cite{chitambar2015assisted} based on arguments of \cite{modi2010unified}]\label{geom}
	Let $\{\Pi_c=|c \rangle\langle c|\otimes \id \}$ be the set of projectors on the incoherent basis  ${\cal Z}$ on $A$. Then:
	\begin{equation}
		\Delta_{\cal Z}(\rho)=S(\rho||\sum_c \Pi_c\rho\Pi_c).
	\end{equation}
\end{lem}
\begin{proof}
	We start by observing that $S(\rho||\sigma)$ is a convex function of $\sigma$ (for fixed $\rho$), and hence, the global minimum $\sigma^*$ is the unique stationary point of that function (and it exists). Let us start the calculation of the stationary points,
	by parametrising $\sigma$ by an exponential map: let $\{{\bf o}_k\}$ be a basis for the sub-algebra of (zero trace) hermitian operators in $AB$ satisfying the condition $[\Pi_c,{\bf o}_k]=0$. It is clear that any full rank state $\sigma \in CQ_{\cal Z}$ can be written as $\sigma=\exp(-{\bf h})/Z$, with $Z={\rm Tr}\exp(-{\bf h})$ and ${\bf h}=\sum_k \lambda_k {\bf o}_k$. Notice that non-full rank states can be reached as a limit. In this parametrisation the stationary conditions reduce to
	$$
	{\rm Tr }(\rho {\bf o}_k)-{\rm Tr }(\sigma^* {\bf o}_k)=0.
	$$
	This implies that for the global minimum $\sigma^*$, and for any observable ${\bf O}$, ${\rm Tr }[\rho \sum_c\Pi_c {\bf O}\Pi_c]={\rm Tr }[\sum_c\Pi_c\rho \Pi_c {\bf O}]$ has to equal ${\rm Tr }[\sigma^* {\bf O}]$. Therefore, the global minimum is $\sigma^*=\sum_c\Pi_c\rho \Pi_c$.
\end{proof}

{\bf \em Properties of $\Delta_{\cal Z}$ --- } Notice that from Lemma~\ref{geom} it follows that  the monotone $\Delta_{\cal Z}$ is the same as the recoverable coherence in the basis $\cal Z$ by incoherent operations~\cite{WinterY2015}
$
\Delta_{\cal Z}(\rho_A\otimes\rho_B)=C_{\cal Z}(\rho_A)
$ and that the monotone is additive:
\begin{lem}[Additivity]\label{lem:additivity}
	\begin{equation}
		\Delta(\rho^n)=n \Delta(\rho).
	\end{equation}
\end{lem}
\begin{proof}
	\begin{eqnarray*}
		\Delta(\rho^{\otimes n})&=& S(\rho^{\otimes n}|| \sum_{c_1,c_2,\ldots c_N}\Pi_{c_1,c_2,\ldots,c_n}\rho^{\otimes n}\Pi_{c_1,c_2,\ldots,c_n})\\
		&=& S\left(\rho^{\otimes n}|| \left(\sum_{c}\Pi_{c}\rho\Pi_{c}\right)^{\otimes n}\right)\\
		&=& n S\left(\rho|| \left(\sum_{c}\Pi_{c}\rho\Pi_{c}\right)\right)=n \Delta(\rho).
	\end{eqnarray*}
\end{proof}

We can also verify that $\Delta_{\cal Z}$ is convex:
\begin{lem}[Convexity]\label{lem:convexity}
	\begin{equation}
		\Delta_{\cal Z}(\sum_i p_i \rho_i)\leq \sum_i p_i \Delta_{\cal Z}( \rho_i)\,.
	\end{equation}
\end{lem}
\begin{proof}
	For each $i$, $\Delta_{\cal Z}(\rho_i)=S(\rho_i||\sigma_i)$ for certain $\sigma_i\in CQ_{\cal Z}$, therefore
	\begin{eqnarray*}
		\Delta_{\cal Z}(\sum_i p_i \rho_i)&=& \min_{\sigma \in CQ_{\cal Z}} S(\sum_i p_i \rho_i||\sigma)\\
		&\le& S(\sum_i p_i \rho_i||\sum_i p_i\sigma_i)\\
		&\le&  \sum_i p_i S(\rho_i||\sigma_i).
	\end{eqnarray*}
\end{proof}

%%%%%%%%%%%%%%%%%%%%%%%%%%%%%%%%%%%%%%%%%%                                 %%%%%%%%%%%%%%%%%%%%%%%%%%%%%%%%%%%%%%%%%%%%%%%%%%%%%%
Now, we can show that this quantity is non-increasing under $GOIA_{\cal Z}$:
\begin{prop}[Monotonicity on average]\label{Prop:monotone}
	\begin{equation}
		\Delta_{{\cal Z}}(\Lambda(\rho))\leq \Delta_{{\cal Z}}(\rho)\;\;\; \forall \Lambda \in GOIA_{\cal Z}, \;CPTP.
	\end{equation}
\end{prop}

\begin{proof}
	We start noticing that if $\Lambda$ is in $GOIA_{\cal Z}$, $\Lambda(CQ_{\cal Z})\subset CQ_{\cal Z}$. Therefore,
	\begin{eqnarray*}
		\Delta_{\cal Z}(\Lambda(\rho))&=&\min_{\sigma \in CQ_{\cal Z}} S(\Lambda(\rho)||\sigma)\\
		&\leq&\min_{\sigma \in \Lambda(CQ_{\cal Z})} S(\Lambda(\rho)||\sigma)\\
		&=&\min_{\sigma \in CQ_{\cal Z}} S(\Lambda(\rho)||\Lambda(\sigma))\\
		&\leq&\min_{\sigma \in CQ_{\cal Z}} S(\rho||\sigma)= \Delta_{\cal Z}(\rho),
	\end{eqnarray*}
	where the last inequality follows from the monotonicity of the relative entropy under CPTP maps~\cite{lindblad.74}.
\end{proof}

{\em \bf Proof of the geometric  upper bound (Eq. 7)---\phantomsection{}\addcontentsline{toc}{section}{Upper bound}}\label{upper bound} Here we present the proof of Eq. 7, stating that the geometric monotone defined by the relative entropy to the free states, upper bounds $C_{\cal Z}^{REC}$.

\begin{proof}
	Let $\Lambda_n$ be the optimal $GOIA_{\cal Z}$ map that produces from the state $\rho$, $m_n$ copies of the state $\ket{\Psi}$, $\epsilon$-near in fidelity to $\ket{+}$:
	$$
	\Lambda_n(\rho^n)=|\Psi\rangle\langle \Psi|^{\otimes m_n}.
	$$
	Now we observe that up to order $\epsilon$ (using the continuity of the fidelity and of the von Neumann entropy as functions of the state):
	\begin{eqnarray*}
		m_n&=&  S\left({\rm Tr}_R \Lambda(\rho^{\otimes n})\| {\rm Tr}_R \sum_c \Pi_c \Lambda(\rho^{\otimes n})\Pi_c \right)\\
		&\leq& S \left(\Lambda(\rho^{\otimes n})|\ \sum_c \Pi_c \Lambda(\rho^{\otimes n})\Pi_c\right)\\
		&=& \Delta_{\cal Z}(\Lambda(\rho^{\otimes n}))\\
		&\leq& \Delta_{\cal Z}(\rho^{\otimes n})\\
		&=& n \Delta_{\cal Z}(\rho).
	\end{eqnarray*}
	The first line follows from the definition of the relative entropy and the second from the monotonicity of this quantity under the partial trace. The third line follows from the definition of $\Delta_{\cal Z}$, the fourth from the monotonicity under $GOIA_{\cal Z}$ and the last one from additivity (Lemma~\ref{lem:additivity}).
	
	Now, from the definition of $C^{REC}_{{\cal Z}}(\rho)$ (Definition 4) we obtain
	$$
	C^{REC}_{{\cal Z}}(\rho) =
	\lim_{\epsilon\rightarrow\infty}\lim_{n\rightarrow \infty} \max_{\Lambda \in GOIA_{\cal Z}}  \frac{m(\Lambda,\rho^n)}{n}\leq
	\Delta_{{\cal Z}}(\rho).
	$$
\end{proof}

{\em \bf Lower bounds ---\phantomsection{}\addcontentsline{toc}{section}{Lower bounds}}\label{lower_bounds} A lower bound for $C_{\cal Z}^{REC}$ is provided by looking at the final coherence obtained after a specific protocol. A subfamily of such protocols consists on performing a measurement on the $B$ side, communicating the outcome,
and adding a label to the classical side:
$$
 \Lambda(\rho_{AB})= \sum_m  |k\rangle\langle k| \otimes ((\id\otimes M_k) \rho_{AB} (\id\otimes M_k)^\dagger),
$$
where $\{M_k\}$ defines a POVM. Tracing out $B$ we obtain for the relative entropy of coherence of the final state:
$$
C_{\cal Z}(\tr_B \Lambda(\rho))=\sum_k p_k C_{\cal Z}(\rho_k),
$$
with $p_k=\tr (M_k^\dagger M_k \rho)$, $\rho_k= |k\rangle\langle k| \otimes\tr_B (M_k^\dagger M_k \rho)$. We can  also see that this is the maximum amount of coherence that can be recovered by local means (i.e. in the $LQICC_{\cal Z}$ framework). Noticing that $C_{\cal Z}(\tr_B \Lambda(\rho))$ is a concave function on the set of POVMs on $B$ we can see that its maximum is attained on the boundary of the set. Following a similar reasoning as in the optimization of discord-like quantities \cite{Ham.04}, we can reduce the optimization problem to find a set of rank-1 projectors. Moreover, if  $\rank(\tr_B \rho)=r$, the maximum is attained for a POVM with at most $r^2$ elements.\\
As a corollary, we can notice that for the set of ``Quantum-Classical states'' $\rho=\sum_{k} p_k \rho_{k}^{A}\otimes |k\rangle \langle k|_B$ the lower bound coincides with the upper bound $\Delta_{\cal Z}(\rho)$ and hence, $C_{\cal Z}^{REC}(\rho)=\Delta_{\cal Z}(\rho)$.

{\em Non equivalence for $C_{\cal Z}^{REC}(\rho)$ and $\Delta_{\cal Z}(\rho)$ in the general case.} 
An open question about $C_{\cal Z}^{REC}$ is related to its numerical equivalence with the geometric measure for general mixed states. To illustrate the problem, let us consider the mixed state
$$
\rho= \frac{|\phi\rangle\langle \phi|_A}{2}\otimes |0\rangle \langle 0|_B + 
\frac{|\varphi\rangle\langle \varphi|_A}{2}\otimes |+\rangle \langle +|_B
$$
with $|\phi\rangle=\frac{|0\rangle +|1\rangle}{\sqrt{2}}$ and $|\varphi\rangle=\frac{|1\rangle +|2\rangle}{\sqrt{2}}$. For this state, $\Delta_{\cal Z}(\rho)\approx 0.8925$ while $C_{\cal Z}({\rm Tr}_B\rho)\approx  0.6887$. A better lower bound is given by the previous local protocol involving the optimal measurement on $B$, followed by adding an ancilla on $A$. By numerical optimisation over projective measurements on $B$  a lower bound of $C_{\cal Z}({\rm Tr}_B\Lambda_{loc}\rho)\approx  0.8167$ was obtained. This number is the best we can obtain by local means in a single shot protocol. However, to evaluate $C_{\cal Z}^{REC}(\rho)$, we should exhaust every protocol using infinite many copies, which would be possible only by providing a specific upper bound for $\rho$, and a protocol that saturates it.     

\begin{lem}[Recoverable coherence for pure states, also see~\cite{streltsov2015hierarchies}]\label{app:RCpurestates}
If $\rho_{AB}=| \psi \rangle \langle \psi |$  $C_{\cal Z}^{REC}(\rho_{AB})=\Delta_{\cal Z}(\rho_{AB})$.
\end{lem}
\begin{proof}
	 We start by noticing that, due to the Schmidt decomposition theorem $|\Psi\rangle=\sum_{\alpha} \lambda_\alpha |\alpha\rangle_A |\alpha\rangle_B$ for certain constants $\{\lambda_{\alpha}\}$ and local orthogonal basis $\{|\alpha\rangle_A\}$,  $\{|\alpha\rangle_B\}$.   We can bring this state to a locally incoherent maximally correlated state $|\psi'\rangle$ by adding an ancilla on $B$ in a reference state $|0\rangle_{B'}$, followed by the application of a controlled translation operation ${\bf T}_c=\sum_i|i\rangle\langle i|_{A}\otimes{\bf 1}_B\otimes \sum_k |i\oplus k\rangle\langle k|_{B'}$. Since both ${\bf T}_c$ and its inverse are free operations (being controlled unitaries with incoherent control), by the monotonicity of the upper bound under $GOIA_{\cal Z}$-operations we get that, $\Delta_{\cal Z}(|\psi'\rangle\langle \psi'|)=\Delta_{\cal Z}({\bf T}_c(\rho_{AB}\otimes |0\rangle \langle 0|){\bf T}_c^\dagger)=\Delta_{\cal Z}(\rho_{AB}\otimes |0\rangle \langle 0|)=\Delta_{\cal Z}(\rho_{AB})$. It follows from the tightness of the bound for maximally correlated states that
	 $\Delta_{\cal Z}(\rho_{AB}) = \Delta_{\cal Z}(|\psi'\rangle\langle \psi'|)=C_{\cal Z}^{REC}(|\psi'\rangle\langle \psi'|)\leq C_{\cal Z}^{REC}(|\psi\rangle\langle \psi|) \leq \Delta_{\cal Z}(\rho_{AB})$, and the equality is shown. 
\end{proof}

 {\em \bf  Proof of Eq. \ref{eq:precvscoh} --- \phantomsection{}\addcontentsline{toc}{section}{Coherence and Accurateness in DQC1}}\label{app:precdqc1} 
In this section we prove Eq. \ref{eq:precvscoh} of the main text, 
	\begin{equation*}
		prec(\frac{\tr {\bf U}}{\rm dim}) \approx -\log_2 \left({\SE}\left( \frac{\tr {\bf U}}{\dim
		}\right)\right) \approx    \frac{1}{2}\log_2(C_{\cal Z}^{REC}(\rho_{0}))\;,
		\end{equation*}
We find it instructive to first consider a slight generalisation of the DQC1 protocol, where instead of a source of maximally coherent probes, a general state $\rho_{AB}^{\otimes m}$ is provided (See Figure \ref{fig:dqc1rc}). From this state, (in the asymptotic limit) we can prepare
$n\approx C_{\cal Z}^{REC}(\rho_{AB}^{\otimes m})= m C_{\cal Z}^{REC}(\rho_{AB})$ maximally coherent probes. In DQC1, the estimation of the trace is given by 
$
\frac{{\rm Tr}{\bf U}}{dim}=\langle \sigma_x\rangle + {\bf i}\langle \sigma_y\rangle
$, being 
$\sigma_x=(^0_1\,^1_0)$, $\sigma_y=(^0_{{\bf i}}\,^{-{\bf i}}_0)$  the $x$, $y$ Pauli matrices and $\langle \ldots \rangle$ the expectation values. If the number of available probes $n$ is large, $\langle \ldots \rangle$ can be approximated by the average over the results of the outcomes of the  $n$ \emph{independent} runs of the algorithm $\langle \ldots \rangle_n$. In the asymptotic limit, the error introduced by replacing $\langle \ldots \rangle$ by $\langle \ldots \rangle_n$ is given by the Standard Error of the Mean
$\SE[\sigma_\mu]=\sqrt{\frac{\langle \sigma_\mu^2 \rangle-\langle\sigma_\mu\rangle^2}{n}}=\sqrt{\frac{1-\langle \sigma_\mu\rangle ^2}{n}}$  
with probability $\geq 68\%$\cite{devore.08}. From this it is straightforward to see that in the asymptotic limit, the precision (i.e. the number of significant (binary) digits) of a number $x$, $|x|<1$  goes like:
$$
prec(\hat{x})\approx-\log_2(\SE(\hat{x})).
$$ 
In the same way, we can bound the norm of the error in the estimation of $\frac{\tr {\bf U}}{\dim}$ by $\SE\left(\frac{\tr {\bf U}}{\dim}\right)\approx \sqrt{\SE(\sigma_x)^2+\SE(\sigma_y)^2}=\sqrt{2-(|\frac{\tr {\bf U}}{\dim}|)^2}/\sqrt{n}=
\sqrt{2-(|\frac{\tr {\bf U}}{\dim}|)^2}/{\sqrt{C_{\cal Z}^{REC}(\rho_{AB}^{\otimes m})}}$. But $1 \leq \sqrt{2-(|\frac{\tr {\bf U}}{\dim}|)^2} \leq \sqrt{2}$ and hence  
$prec\approx -\log_2(\SE({\rm Tr}{\bf U}/{dim})) \approx \frac{1}{2}\,\log_2C_{\cal Z}^{REC}(\rho_{AB}^{\otimes m})$, using this generalized algorithm.

\begin{figure}
	\centering
	\scalebox{.2}{\includegraphics{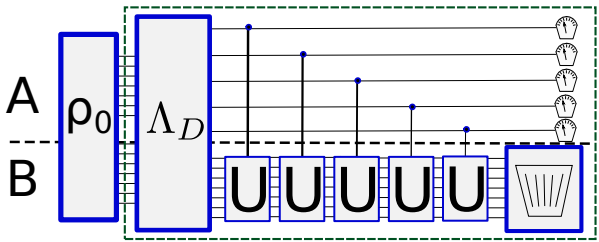}}
	\caption{{\bf Recoverable Coherence and  DQC1.} In this figure, the finite accuracy DQC1 algorithm, including the purification stage is depicted.
		Starting from a general resource state $\rho_0$, the channel $\Lambda_D$ prepares with some fidelity the state $(|+\rangle\otimes \frac{\bf 1}{{\rm Tr}{\bf 1}})^{m}$, which is used to estimate the trace of $U$.
	}
	\label{fig:dqc1rc}
\end{figure}
To see that this is indeed the maximum attainable precision for  the standard algorithm without the recovering step, let us consider  $m$ probes in the general state $\rho_{\rm probe}=\left(^p_{\alpha}\;_{1-p}^\alpha\right)$, where without loss of generality we assume $\alpha\geq 0$. For these probes, $\frac{{\rm Tr}{\bf U}}{dim }=\frac{\langle \sigma_x\rangle + {\bf i} \langle \sigma_y\rangle}{\alpha}$ while $\SE\left (\frac{{\rm Tr}{\bf U}}{\dim }\right)= \sqrt{\frac{2 - \alpha^2 \left |\frac{{\rm Tr}{\bf U}}{\dim }\right|^2}{\alpha^2 m}}$. Thus, 
$$
prec_\alpha\approx-\log_2(|\SE \frac{\tr {\bf U}}{\dim}|)\approx -\log_2 \frac{\sqrt{2-|\alpha|^2|\frac{\tr {\bf U}}{\dim}|^2}}{\sqrt{\alpha^2 m}}.
$$
But $m=C_{\cal Z}^{REC}(\rho_{\rm probe}^{\otimes m})/C_{\cal Z}^{REC}(\rho_{\rm probe})$ and hence,
$$
prec_\alpha\approx\frac{1}{2}\log_2(C_{\cal Z}^{REC}(\rho_{\rm probe}^{\otimes m}))
+\log_2\sqrt{\frac{\alpha^2/C_{\cal Z}^{REC}(\rho_{\rm probe})}{(2-|\alpha|^2|\frac{\tr {\bf U}}{\dim}|^2)}}.
$$
But the second term is a number less or equal than $0$ since $\alpha^2\leq C_{\cal Z}(\rho_{\rm probe})= C_{\cal Z}^{REC}(\rho_{\rm probe})$. As this term is finite and independent of $m$, the correction can be neglected in the asymptotic limit. 
{\phantomsection{}\addcontentsline{toc}{part}{References}}\label{references}

\end{document}